\documentclass[journal,onecolumn,12pt, draftclsnofoot]{IEEEtran}
\usepackage{amsfonts,color,morefloats}
\usepackage{booktabs}
\usepackage{amssymb,amsmath,latexsym,amsthm}
\usepackage{graphicx}
\usepackage{makecell}
\usepackage{hyperref}
\usepackage{multirow}
\usepackage{array}
\usepackage{algorithm}
\usepackage{algorithmic}
\usepackage{float}

\newtheorem{theorem}{Theorem}[section]
\newtheorem{lemma}[theorem]{Lemma}

\newtheorem{remark}{Remark}

\newtheorem{definition}{Definition}
\newtheorem{example}{Example}

\usepackage{enumerate}

\begin{document}

\title{A Criterion to Determine True Minimum Distances of Goppa Codes
	\thanks{The work of Shuying Dong and Chengju Li was supported by the National
	Natural Science Foundation of China (T2322007) and the State Key Laboratory of Cryptography and Digital Economy Security (KFYB2502). 
	The work of Hao Chen was supported by the National
	Natural Science Foundation of China (62032009).
	(Corresponding author: Chengju Li.)}
}

\author{Shuying Dong ~~ Hao Chen  ~~ Yaqi Chen ~~  Ziyan Xie ~~  Chengju Li\\
	\thanks{S. Dong is with the College of Information Engineering, Shanghai Maritime
		University, Shanghai, 201306, China (email: shying2026@163.com).}
	\thanks{H. Chen is with the College of Information
		Science and Technology, Jinan University, Guangzhou, Guangdong, 510632, China (email: haochen@jnu.edu.cn).}
	\thanks{Y. Chen is with the College of Cyber Security, Jinan University, Guangzhou, Guangdong, 510632, China (email: chenyq@stu.jnu.edu.cn).}
	\thanks{Z. Xie is with the School of Mathematics, Nanjing University of Aeronautics and Astronautics, Nanjing, 211106, China (email: xieziyan@nuaa.edu.cn).}
		\thanks{C. Li is with the MoE Engineering Research
		Center of Software/Hardware Co-Design Technology and Application, East China Normal University, Shanghai, 200062, China, and with State Key Laboratory of Cryptography and Digital Economy Security, Shandong University, Shandong, 266237, China (email: cjli@sei.ecnu.edu.cn).}
}

\date{\today}

\maketitle

\begin{abstract}
	Goppa codes play an important role in code-based cryptography due to their efficient decoding algorithms and their use as underlying private
	codes in the McEliece cryptosystem. To determine the true minimum  distances of Goppa codes is a notoriously difficult problem.
	In this paper, we establish a criterion for a
	Goppa code to attain its designed distance. We consider Goppa
	polynomials of the form
	$
	G(x)=U(x)H(x)+V(x)H'(x),
	$
	where $\deg(G)=t$ and $H(x)\in\mathbb{F}_q[x]$ is a monic irreducible
	polynomial of degree $t+1$ whose roots are contained in the support $L$.
	We prove that the corresponding Goppa code $\Gamma(L,G)$ contains a codeword of weight $t+1$
	 if and only if
	\[
	\frac{V(\alpha_{i_{t+1}})}{V(\alpha_{i_j})}\in\mathbb{F}_q^*,
	\qquad 1\leq j\leq t,
	\]
	where $\alpha_{i_1},\ldots,\alpha_{i_{t+1}}$ are the roots of $H(x)$.
	Based on this criterion, we derive a general family of Goppa codes that attain their designed distance by developing an interpolation-based construction of Goppa polynomials. We further obtain families of Goppa codes whose Goppa polynomials are determined by considering
	monomial, binomial, and their product of the auxiliary polynomial $V(x)$.
	By taking $H(x)$ to be different irreducible binomials and trinomials, we obtain several explicit families of Goppa codes whose minimum distances are equal to designed distance.
\end{abstract}

{\bf  Keywords} Goppa code, minimum distance.

\section{Introduction} \label{sec-intro}

Goppa codes were first introduced by V. D. Goppa in the early 1970s \cite{Goppa1970,Goppa1972}. 
As an important subclass of alternant codes, Goppa codes have been extensively studied due to their applications in algebraic coding theory and code-based cryptography. In particular, certain families of Goppa codes asymptotically attain the Gilbert-Varshamov bound \cite{HP2003} and several well-known algebraic codes, including BCH codes and Srivastava codes, are closely related to Goppa codes under suitable choices of parameters \cite{MS1977}.
  Moreover, the algebraic structure of Goppa codes enables efficient decoding algorithms, such as the Berlekamp-Massey algorithm \cite{Berlekamp1973}, the Patterson decoding algorithm \cite{Patterson1975}, and the modified extended Euclidean algorithm proposed by Sugiyama et al. \cite{Sugiyama1975}. 
	More importantly, the McEliece cryptosystem  \cite{McEliece1978} based on binary Goppa codes is considered a promising candidate for post-quantum cryptography due to its resistance against known quantum attacks.
  Consequently, extensive efforts have been devoted to investigating the structures, properties, and parameters of Goppa codes.

Let $\Bbb F_q$ be the finite field of order $q$, where $q$ is a prime power.
For a vector $\mathbf{c}=(c_0,\ldots,c_{n-1})\in\mathbb{F}_q^n$, let $\operatorname{supp}(\mathbf{c})=\{i:c_i\neq0\}$ and $\operatorname{wt}(\mathbf{c})=|\operatorname{supp}(\mathbf{c})|$.
The minimum distance of a linear code $\mathcal{C}\subseteq\mathbb{F}_q^n$ is
$$
d(\mathcal C)
=
\min_{\mathbf c\in\mathcal C\setminus\{\mathbf 0\}}
\operatorname{wt}(\mathbf c).
$$
A $q$-ary linear code with parameters $[n,k,d]$ is a $k$-dimensional subspace of $\mathbb{F}_q^n$ with minimum distance $d$.
Let $\mathbb F_{q^m}$ be an extension of
$\mathbb F_q$. 
\begin{definition}\label{lem-goppa}\cite{Goppa1970}
	Let $L= \{\alpha_1,\ldots,\alpha_{n}\} \subseteq \Bbb F_{q^m}$ be a set of $n$ distinct elements and $G(x) \in \Bbb F_{q^m}[x]$ be a polynomial of degree $t\le n$ such that $G(\alpha_i) \ne 0$ for $i=1,2,\ldots, n$. The $q$-ary Goppa code $\Gamma(L,G)$ with Goppa polynomial $G(x)$ and support $L$ is defined by
	$$
	\Gamma(L,G) = \biggl\{(c_1,c_2,\ldots,c_{n})\in \Bbb F_q^n \ \bigg| \sum_{i=1}^{n} \frac{c_i}{x-\alpha_i} \equiv 0 \pmod {G(x) } \biggr\}.
	$$
	If $G(x)$ is irreducible over $\Bbb F_{q^m}$, then $\Gamma(L,G)$ is called irreducible Goppa code.
\end{definition}

\begin{lemma}\label{lem-bound}\cite{Goppa1970}
	Let $\Gamma(L,G)$ be a $q$-ary Goppa code with $\deg(G) = t$. Then $\Gamma(L,G)$ is an $[n,k,d]$ linear code
	satisfying
	$
	k\geq n-mt$ and $d\geq t+1.$ The value $\delta=t+1$ is called its designed distance.
	If $q=2$ and $G(x)$ is square-free, then
	$
	d\geq 2t+1.
	$
\end{lemma}

Although the Goppa bounds provide lower bounds on the dimensions and minimum distances of Goppa codes, determining their exact values remains a difficult problem.
Early studies mainly focused on improving general bounds on the dimension and minimum distance of Goppa codes. In 1976, Sugiyama et al. derived improved bounds on the parameters of Goppa codes \cite{Sugiyama1976}. Feng subsequently improved lower bounds on the minimum distance and upper bounds on the number of parity-check symbols \cite{Feng1983}, while Yue and Ji studied minimum distance bounds for Goppa and alternant codes \cite{YueJi1991}, and Yue obtained further results on the dimension of Goppa codes \cite{Yue1991}.
 Later work increasingly focused on structured families of Goppa codes whose parameters could be bounded more tightly or determined exactly.
 Bezzateev and Shekhunova constructed a subclass of binary Goppa codes 
 defined by the Goppa polynomial
 $
 G(x)=x^t+A,
 $
 where $t\mid 2^m-1$ and $A$ is a $t$-th power in 
 $\mathbb{F}_{2^m}^{*}$. They showed that these codes attain their 
 designed distance \cite{Bezzateev1995}. Subsequently,  they determined
 the minimum distances of certain chains of separable binary Goppa codes
 \cite{Bezzateev2008}, and derived improved parameter estimates for
 special separable Goppa codes \cite{Bezzateev2010}. In 1988, V\'eron studied
 trace Goppa codes and obtained improved bounds on the dimension and
 minimum distance by exploiting the algebraic properties of the trace
 operator \cite{Veron1998}. He subsequently determined the true dimensions of certain binary quadratic trace Goppa codes and proved several related dimension conjectures \cite{Veron2001,Veron2005}. Couvreur et al. established structural identities for wild Goppa codes \cite{Couvreur2016}. More quasi-cyclic and trace Goppa codes were studied to obtain improved parameter bounds and new optimal or best-known codes in \cite{Bezzateev2017}, \cite{Byrne2023}, and the references therein. More recently, Wu et al. determined the exact minimum distance of a family of binary Goppa codes defined by $x^{3t}+1$ \cite{Wu2024}. However, a general and direct criterion for determining when a Goppa code attains its designed distance is still lacking.

In this paper, we mainly establish a criterion for a Goppa code 
$\Gamma(L,G)$ to attain its designed distance, where the 
Goppa polynomial is of the form
$
G(x)=U(x)H(x)+V(x)H'(x),
$
with $\deg(G)=t$ and $H(x)\in\mathbb{F}_q[x]$ being a monic irreducible
polynomial of degree $t+1$ whose roots are contained in the support
$L$.
By analyzing the parity-check matrix of $\Gamma(L,G)$, we transform the existence of a codeword of weight $t+1$ into a condition on the values of $V(x)$ at the conjugate roots of $H(x)$. More precisely, we prove that $\Gamma(L,G)$ contains a codeword $\mathbf{c}\in\mathbb F_q^n$ with $\operatorname{wt}(\mathbf{c}) = t+1$ if and only if
$$
\frac{V(\alpha_{i_{t+1}})}{V(\alpha_{i_j})}\in\mathbb F_q^*,\qquad 1\le j\le t,
$$
where $\alpha_{i_1},\ldots,\alpha_{i_{t+1}}$ are the roots of $H(x)$.
Consequently, whenever this condition is satisfied, $d(\Gamma(L,G))=t+1$ and the designed distance bound is achieved.
Furthermore, we derive several families of Goppa codes based on the criterion above, which are summarized as
follows.

\begin{itemize}
	\item By prescribing the values of $V(x)$ on the roots of $H(x)$,  we derive a general family of Goppa codes that attain their designed distance by developing an interpolation-based construction of Goppa polynomials.
	\item  By considering monomial, binomial, and their product forms of the auxiliary
	polynomial $V(x)$, we obtain families of Goppa
	codes that attain their designed distance, with the corresponding Goppa
	polynomials computed through
	$
	G(x)\equiv V(x)H'(x)\pmod{H(x)}.
	$
\end{itemize}
 In particular, by taking specific irreducible polynomials $H(x)$, including
binomials $x^{t+1}-a$ and trinomials $x^p-x-a$, we obtain several explicit families of Goppa
codes with exactly determined minimum distances.

This paper is organized as follows. Section~\ref{sec-pre} introduces the notation used throughout the paper. Section~\ref{sec-criterion}
establishes the minimum distance criterion and derives interpolation-based families of Goppa polynomials. Section~\ref{sec-V} develops constructions of Goppa codes with exactly determined minimum distances from the auxiliary polynomial $V(x)$. Section~\ref{sec-sum} concludes the paper. 

In this paper, by the Database, we mean the tables of best-known linear codes
\cite{G}, which are maintained by Markus Grassl at http://www.codetables.de/.

\section{Preliminaries} \label{sec-pre}

Throughout this paper, let $\mathbb F_q$ denote the finite field of order
$q$, where $q$ is a power of a prime $p$, and let
$\mathbb F_q^*=\mathbb F_q\setminus\{0\}$ denote its multiplicative
group. Let $\mathbb F_{q^m}$ be an extension field of $\mathbb F_q$.
In this section,  we introduce
the notation used throughout the paper.

Let
$
H(x)\in\mathbb F_q[x]
$
be a monic irreducible polynomial of degree $t+1$. A polynomial over a field is
called separable if it has no repeated roots in its splitting field. Equivalently,
a polynomial $f(x)$ is separable if and only if
$
\gcd(f(x),f'(x))=1,
$
where $f'(x)$ denotes the formal derivative of $f(x)$.
A field is called perfect if every irreducible polynomial over it is
separable. Every finite field is perfect because the Frobenius map
$$
\varphi:\mathbb F_q\longrightarrow\mathbb F_q,
\qquad
a\longmapsto a^p,
$$
is an automorphism. Therefore, the irreducible polynomial $H(x)$ is
separable and 
$
\gcd(H(x),H'(x))=1.
$
Then we have
$
H'(\theta)\neq0
$
for every root $\theta$ of $H(x)$.
 Since $H(x)\in \Bbb F_q[x]$ is irreducible with $\deg(H) = t+1$, we have
$$
\mathbb F_q(\theta)\cong\mathbb F_{q^{t+1}}.
$$
Moreover, the roots of $H(x)$ are the pairwise distinct Frobenius conjugates of $\theta$, namely
$\theta,\theta^q,\theta^{q^2},\ldots,\theta^{q^t}$.
In the remainder of this paper, we consider Goppa polynomials of the form
\begin{equation}\label{eq-Gx}
	G(x)=U(x)H(x)+V(x)H'(x),
\end{equation}
where
$
U(x),V(x)\in\mathbb F_{q^m}[x]
$
and
$
\deg(G)=t.
$

To obtain explicit families of Goppa polynomials, we consider specific
choices of the irreducible polynomial $H(x)$ appearing in \eqref{eq-Gx}. The following two lemmas
provide conditions under which the binomial $x^s-a$ and the trinomial
$x^p-x-a$ are irreducible over $\mathbb{F}_q$, and will be used
throughout this paper.

\begin{lemma}\label{lem-binomial-irreducible}
	Let $s\geq 2$ be an integer and $a\in \mathbb{F}_q^*$.
	Let $e$ be the order of $a$ in $\mathbb{F}_q^*$. Then the binomial
	$x^s-a$ is irreducible in $\mathbb{F}_q[x]$ if and only if the following
	two conditions are satisfied:
	\begin{enumerate}
		\item each prime factor of $s$ divides $e$, but does not divide
		$(q-1)/e$;
		\item $q\equiv 1 \pmod 4$ if $s\equiv 0 \pmod 4$.
	\end{enumerate}
\end{lemma}

\begin{lemma}\label{lem-trinomial-irreducible}
	Let $a\in \mathbb{F}_q$ and let $p$ be the characteristic of $\Bbb F_q$. Then the trinomial
	$
	x^p-x-a
	$
	is irreducible over $\mathbb{F}_q$ if and only if
	$
	\operatorname{Tr}_{\mathbb{F}_q/\mathbb{F}_p}(a)\neq 0.
	$
\end{lemma}

\section{A Minimum Distance Criterion for Goppa Codes}\label{sec-criterion}
Throughout this section, we assume that $t+1\mid m$ and $H(x)\in \Bbb F_q[x]$ is a monic irreducible polynomial with $\deg(H) = t+1$.
 Let $L=\{\alpha_1,\ldots,\alpha_n\} \subseteq\mathbb F_{q^m}$ be the support of the Goppa code, and suppose that all roots of $H(x)$ are contained in $L$. Thus there exist pairwise distinct indices $i_1,\ldots,i_{t+1}\in\{1,\ldots,n\}$ such that $\alpha_{i_1},\ldots,\alpha_{i_{t+1}}$ are exactly the roots of $H(x)$. We further assume that the Goppa polynomial $G(x)$ is given by \eqref{eq-Gx}. Since $H(\alpha_{i_j})=0$ for every $1\le j\le t+1$, we have $G(\alpha_{i_j})=V(\alpha_{i_j})H'(\alpha_{i_j})$. Moreover, it follows that $V(\alpha_{i_j})\neq0$ for $1\leq j\leq t+1$  since $G(\alpha_{i_j})\neq0$ and
 $H'(\alpha_{i_j})\neq0$.
In the following theorem, we establish a criterion for $\Gamma(L,G)$ to attain its designed distance. More precisely, we give a necessary and sufficient condition for $\Gamma(L,G)$ to contain a codeword of weight $t+1$.

\begin{theorem}\label{thm-1}
	Let $\Gamma(L,G)$ be a Goppa code over $\mathbb F_q$ with $\deg (G)=t$, where $G(x)$ is given by \eqref{eq-Gx}. Assume that $\alpha_{i_1},\ldots,\alpha_{i_{t+1}}$ are exactly the roots of $H(x)$ in $L$. Then $\Gamma(L,G)$ contains a codeword $\mathbf{c}\in\mathbb F_q^n$ with $\operatorname{supp}(\mathbf{c})=\{i_1,\ldots,i_{t+1}\}$ if and only if
	\begin{equation}\label{eq-v-condition}
		\frac{V(\alpha_{i_{t+1}})}{V(\alpha_{i_j})}\in\mathbb F_q^*,\qquad 1\le j\le t.
	\end{equation}
	Consequently, if \eqref{eq-v-condition} holds, then $d(\Gamma(L,G))=t+1$.
\end{theorem}

\begin{proof}
	By the Goppa bound, we have $d(\Gamma(L,G))\ge t+1$. Hence, to prove that $d(\Gamma(L,G))=t+1$, it suffices to construct a codeword of weight $t+1$.
	
	Let $\mathbf{c}=(c_1,\ldots,c_n)\in\mathbb F_q^n$ be a vector whose nonzero positions are $\{i_1,\ldots,i_{t+1}\}$. It follows from the parity-check matrix of Goppa code that $\mathbf{c}\in\Gamma(L,G)$ if and only if
	\begin{equation}\label{eq-linear-system}
		\sum_{\ell=1}^{t+1}c_{i_\ell}\frac{\alpha_{i_\ell}^r}{G(\alpha_{i_\ell})}=0,\qquad 0\le r\le t-1.
	\end{equation}
	We then separate the last coordinate $c_{i_{t+1}}$. Let
	$$
	\mathbf{V}=
	\begin{pmatrix}
		1 & 1 & \cdots & 1\\
		\alpha_{i_1} & \alpha_{i_2} & \cdots & \alpha_{i_t}\\
		\vdots & \vdots & \ddots & \vdots\\
		\alpha_{i_1}^{t-1} & \alpha_{i_2}^{t-1} & \cdots & \alpha_{i_t}^{t-1}
	\end{pmatrix}
	$$
	and
	$$
\mathbf{D}=\operatorname{diag}\bigl(G(\alpha_{i_1})^{-1},G(\alpha_{i_2})^{-1},\ldots,G(\alpha_{i_t})^{-1}\bigr).
	$$
	Since $\alpha_{i_1},\ldots,\alpha_{i_t}$ are pairwise distinct and $G(\alpha_{i_j})\ne 0$ for all $1\le j\le t$, both $\mathbf{V}$ and $\mathbf{D}$ are invertible. Therefore, \eqref{eq-linear-system} is equivalent to 
	$$
	\mathbf{VD}
	\begin{pmatrix}
		c_{i_1}\\
		c_{i_2}\\
		\vdots\\
		c_{i_t}
	\end{pmatrix}
	=
	-\frac{c_{i_{t+1}}}{G(\alpha_{i_{t+1}})}
	\begin{pmatrix}
		1\\
		\alpha_{i_{t+1}}\\
		\vdots\\
		\alpha_{i_{t+1}}^{t-1}
	\end{pmatrix}.
	$$
	Then,	
	\begin{equation}\label{eq-coordinates}
		\begin{pmatrix}
			c_{i_1}\\
			c_{i_2}\\
			\vdots\\
			c_{i_t}
		\end{pmatrix}
		=
		-\frac{c_{i_{t+1}}}{G(\alpha_{i_{t+1}})}
		\mathbf{D}^{-1}\mathbf{V}^{-1}
		\begin{pmatrix}
			1\\
			\alpha_{i_{t+1}}\\
			\vdots\\
			\alpha_{i_{t+1}}^{t-1}
		\end{pmatrix}.
	\end{equation}
	By the Lagrange interpolation formula, the $j$-th entry of $\mathbf{V}^{-1}(1,\alpha_{i_{t+1}},\ldots,\alpha_{i_{t+1}}^{t-1})^T$ is
	$$
	\prod_{\substack{\ell=1\\ \ell\ne j}}^t
	\frac{\alpha_{i_{t+1}}-\alpha_{i_\ell}}{\alpha_{i_j}-\alpha_{i_\ell}}.
	$$
	On the other hand, since
	$
	H(x)=\prod_{\ell=1}^{t+1}(x-\alpha_{i_\ell}),
	$
	we have
	$$
	H'(\alpha_{i_{t+1}})=\prod_{\ell=1}^t(\alpha_{i_{t+1}}-\alpha_{i_\ell})
	\text{ and }
	H'(\alpha_{i_j})=(\alpha_{i_j}-\alpha_{i_{t+1}})\prod_{\substack{\ell=1\\ \ell\ne j}}^t(\alpha_{i_j}-\alpha_{i_\ell}).
	$$
	It follows that
	\begin{equation}\label{eq-lagrange-derivative}
		\prod_{\substack{\ell=1\\ \ell\ne j}}^t
		\frac{\alpha_{i_{t+1}}-\alpha_{i_\ell}}{\alpha_{i_j}-\alpha_{i_\ell}}
		=
		-\frac{H'(\alpha_{i_{t+1}})}{H'(\alpha_{i_j})}.
	\end{equation}
	Substituting \eqref{eq-lagrange-derivative} into \eqref{eq-coordinates}, we obtain
	\begin{equation}\label{eq-cij-G}
		c_{i_j}
		=
		c_{i_{t+1}}
		\frac{G(\alpha_{i_j})}{G(\alpha_{i_{t+1}})}
		\frac{H'(\alpha_{i_{t+1}})}{H'(\alpha_{i_j})},
		\qquad 1\le j\le t.
	\end{equation}
	Since $G(\alpha_{i_j})=V(\alpha_{i_j})H'(\alpha_{i_j})$ for every $1\le j\le t+1$, Equation \eqref{eq-cij-G} gives that
	\begin{equation}\label{eq-cij-V}
		c_{i_j}=c_{i_{t+1}}\frac{V(\alpha_{i_j})}{V(\alpha_{i_{t+1}})},\qquad 1\le j\le t.
	\end{equation}
	
	We first prove the sufficiency. Assume that \eqref{eq-v-condition} holds. For any $c_{i_{t+1}}\in\mathbb F_q^*$, define
	$$
	c_{i_j}=c_{i_{t+1}}\frac{V(\alpha_{i_j})}{V(\alpha_{i_{t+1}})},\qquad 1\le j\le t,
	$$
	and set $c_i=0$ for $i\notin\{i_1,\ldots,i_{t+1}\}$. 
	Since $V(\alpha_{i_{t+1}})/V(\alpha_{i_j})\in\mathbb F_q^*$ for $1\le j\le t$, we have $\operatorname{wt}(\mathbf{c})=t+1$.
	One can easily verify that $\mathbf{c}$ satisfies \eqref{eq-linear-system} and belongs to $\Gamma(L,G)$. Thus $d(\Gamma(L,G))\le t+1$. Combining this with the Goppa bound
	$
	d(\Gamma(L,G))\ge t+1
	$
	yields
	$
	d(\Gamma(L,G))=t+1.
	$
	
	Conversely, suppose that $\Gamma(L,G)$ contains a codeword $\mathbf{c}\in\mathbb F_q^n$ with $\operatorname{supp}(\mathbf{c})=\{i_1,\ldots,i_{t+1}\}$. Then $c_{i_1},\ldots,c_{i_{t+1}}\in\mathbb F_q^*$. Since $\mathbf{c}\in\Gamma(L,G)$, its nonzero coordinates satisfy \eqref{eq-linear-system}. Repeating the Vandermonde computation above, we obtain \eqref{eq-cij-V}. Therefore,
	$$
	\frac{V(\alpha_{i_{t+1}})}{V(\alpha_{i_j})}
	=
	\frac{c_{i_{t+1}}}{c_{i_j}}
	\in\mathbb F_q^*,\qquad 1\le j\le t.
	$$
	This completes the proof.
\end{proof}

\begin{remark}
	Theorem~\ref{thm-1} provides a characterization of Goppa codes
	$\Gamma(L,G)$ that attain their designed distance. More precisely, for
	a Goppa code with Goppa polynomial of the form \eqref{eq-Gx} and $\deg(G)=t$, the
	attainment of the designed distance is completely determined by the
	values of $V(x)$ at the roots of $H(x)$. 
	This characterization also provides an explicit way to describe
	Goppa polynomials and their supports such that the corresponding codes
	attain their designed distance.
\end{remark}

Based on Theorem~\ref{thm-1}, we next present a general family of Goppa codes satisfying the above criterion by developing an interpolation-based construction of Goppa polynomials. For simplicity,
in the rest of this section, we denote the pairwise distinct roots of
$H(x)$ contained in $L$ by
$\alpha_1,\alpha_2,\ldots,\alpha_{t+1}$.
If the values of $V(\alpha_i)$ are chosen to satisfy
$$
\frac{V(\alpha_{t+1})}{V(\alpha_i)}\in\mathbb F_q^*,
\qquad 1\leq i\leq t,
$$
then \eqref{eq-v-condition} holds automatically.

\begin{theorem}\label{thm-2}
		Let $H(x)\in \Bbb F_q[x]$ be a monic irreducible polynomial with $\deg(H) = t+1$ and $\alpha_1,\alpha_2,\ldots,\alpha_{t+1}$ be the roots of
	$H(x)$ contained in $L$.
	Let $\omega \in \Bbb F_{q^m}^*$ and $\lambda_i,\beta_i\in \mathbb{F}_q^*$ for $1\leq i\leq t+1$.
	Assume that
	$
	\sum_{i=1}^{t+1}\lambda_i\beta_i\neq 0.
	$
	Define
	$$
	G(x)=\omega \sum_{i=1}^{t+1}\lambda_i\beta_i
	\frac{H(x)}{x-\alpha_i}.
	$$
	Further assume that
	$
	G(\eta)\neq 0
	$ for every $\eta\in L$.
	Then the Goppa code $\Gamma(L,G)$ has minimum distance
	$
	d(\Gamma(L,G))=t+1.
	$
\end{theorem}

\begin{proof}
Assume that $\omega \in \Bbb F_{q^m}^*$ and
$
\lambda_i,\beta_i\in \mathbb{F}_q^*
$
for
$
1\leq i\leq t+1.
$
Let
$$
V(x)=\omega \sum_{i=1}^{t+1}\lambda_i\beta_i L_i(x),
$$
where
$$
L_i(x)
=
\prod_{\substack{1\leq k\leq t+1\\ k\neq i}}
\frac{x-\alpha_k}{\alpha_i-\alpha_k}.
$$
Furthermore, it is clear that 
$
L_i(x)
=
\frac{H(x)}{(x-\alpha_i)H'(\alpha_i)}
$ since
$
H(x)=\prod_{k=1}^{t+1}(x-\alpha_k).
$
By the basic property of the Lagrange polynomials, we have
$$
L_i(\alpha_j)
=
\begin{cases}
	1, & i=j,\\
	0, & i\neq j.
\end{cases}
$$
Hence,
it is easy to check that
$
V(\alpha_j)
=
\omega \sum_{i=1}^{t+1}\lambda_i\beta_i L_i(\alpha_j)
=
\omega \lambda_j\beta_j
$  for every
$
1\leq j\leq t+1,
$
and
$$
\frac{V(\alpha_{t+1})}{V(\alpha_j)}
=
\frac{\lambda_{t+1}\beta_{t+1}}{\lambda_j\beta_j}
\in \mathbb{F}_q^*
$$
since
$
\lambda_i,\beta_i\in \mathbb{F}_q^*
$
for any
$
1\leq i\leq t+1.
$
Thus \eqref{eq-v-condition} holds.
Moreover, from \eqref{eq-Gx}, we have
$$
G(x)\equiv V(x)H'(x)\pmod{H(x)}.
$$
Considering this congruence at each root $\alpha_i$ of $H(x)$ gives that
$$
G(\alpha_i)=V(\alpha_i)H'(\alpha_i),
\qquad 1\leq i\leq t+1.
$$
By the Lagrange interpolation formula, the polynomial
$G(x)$ can be uniquely determined by its values at
$
\alpha_1,\alpha_2,\ldots,\alpha_{t+1}
$ since $\deg (G)\le t$.
Thus
$$
G(x)
=
\sum_{i=1}^{t+1}
G(\alpha_i)L_i(x).
$$
Substituting
$
L_i(x)=\dfrac{H(x)}{(x-\alpha_i)H'(\alpha_i)}
$
and
$
G(\alpha_i)=V(\alpha_i)H'(\alpha_i)
$
into the above identity, we get
$$
G(x)
=
\sum_{i=1}^{t+1}
V(\alpha_i)
\frac{H(x)}{x-\alpha_i}.
$$
Since
$
V(\alpha_i)=\omega \lambda_i\beta_i,
$
we further obtain
$$
G(x)
=
\omega \sum_{i=1}^{t+1}
\lambda_i\beta_i
\frac{H(x)}{x-\alpha_i}.
$$
By assumption, $\omega \ne 0$ and 
$
\sum_{i=1}^{t+1}\lambda_i\beta_i\neq 0.
$
Since $H(x)$ is monic of degree $t+1$, the coefficient of $x^t$ in $G(x)$ is nonzero. Then we have $\deg(G)=t$.
It directly follows from Theorem~\ref{thm-1} that the Goppa code $\Gamma(L,G)$ attains its
designed distance, that is,
$
d(\Gamma(L,G))=t+1.
$
We then get the desired conclusion.
\end{proof}

\begin{remark}
	Theorem~\ref{thm-2} yields a general family of Goppa codes that
	attain their designed distance. By varying $H(x)$, $\omega$, and
	$\lambda_i,\beta_i$ subject to the conditions in this theorem, the
	Goppa polynomials can be explicitly given, and the corresponding codes
	attain their designed distance.
\end{remark}

We then present several examples of Goppa codes obtained from
Theorem~\ref{thm-2}, whose parameters are computed by Magma.

\begin{example}
We illustrate Theorem~\ref{thm-2} with two computational examples. In each
case, $H(x)$ is irreducible over $\mathbb F_q$. If $\alpha$ is a root
of $H(x)$, then its roots are ordered as
$
\alpha_i=\alpha^{q^{i-1}}
$ for $1\le i\le t+1$.
For the given values of $\omega,\lambda_i$ and $\beta_i$, we construct
$
G(x)=\omega\sum_{i=1}^{t+1}\lambda_i\beta_i
\frac{H(x)}{x-\alpha_i}.
$
The support $L$ is chosen to contain all roots of $H(x)$ and
satisfies $G(\eta)\neq 0$ for all $\eta\in L$.
	
	\begin{enumerate}
		\item Let $q=2$ and
		$
		H(x)=x^3+x+1.
		$
		Let $\alpha$ be a root of $H(x)$. Then
		$
		\alpha^3=\alpha+1.
		$
		The roots of $H(x)$ are
		$
		\alpha_1=\alpha,\
		\alpha_2=\alpha^2,\ \text{and }
		\alpha_3=\alpha^4=\alpha^2+\alpha.
		$
		Take
		$
		\omega=\alpha+1,
		$
		$
		(\lambda_1,\lambda_2,\lambda_3)=(1,1,1)
		$, and $
		(\beta_1,\beta_2,\beta_3)=(1,1,1).
		$
		Then
		$$
		G(x)
		=
		(\alpha+1)
		\sum_{i=1}^{3}\lambda_i\beta_i
		\frac{x^3+x+1}{x-\alpha_i}=	(\alpha+1)x^2+\alpha+1
		\in\mathbb F_{2^3}[x].
		$$
		Choose a support $L\subseteq\mathbb F_{2^3}$ such that
		$$
		\{\alpha,\alpha^2,\alpha^4\}\subseteq L
		\quad\text{and}\quad
		G(\eta)\ne 0\quad\text{for all }\eta\in L.
		$$
		For example, we take
		$
		|L|=7.
		$
		The code $\Gamma(L,G)$ has parameters $[7,4,3]$ and is optimal according to the database \cite{G}.
		In particular, $d\left(\Gamma(L,G)\right)=\deg H$.

	\item Let $q=3$ and $ H(x)=x^3+x^2+x+2$. Let $\alpha$ be a root of $H(x)$. Then $ \alpha^3=2\alpha^2+2\alpha+1. $ The roots of $H(x)$ are $ \alpha_1=\alpha$, $ \alpha_2=2\alpha^2+2\alpha+1$, and $ \alpha_3=\alpha^2+1. $ Take $ \omega=\alpha+1$, $ (\lambda_1,\lambda_2,\lambda_3)=(1,1,2)$, and $ (\beta_1,\beta_2,\beta_3)=(1,1,1). $ 
	Then $$ G(x) = (\alpha+1) \sum_{i=1}^{3}\lambda_i\beta_i \frac{x^3+x^2+x+2}{x-\alpha_i}=(\alpha+1)x^2 + 2x + (2\alpha^2+2\alpha) \in\mathbb F_{3^3}[x]. $$
	 Choose a support $L\subseteq\mathbb F_{3^3}$ such that $$ \{\alpha,\alpha^3,\alpha^9\}\subseteq L \quad\text{and}\quad G(\eta)\ne 0\quad\text{for all }\eta\in L. $$ For example, we take $ |L|=20. $ The parameters of $\Gamma(L,G)$ are $[20,14,3]$.
	 In particular, $d\left(\Gamma(L,G)\right)=\deg H$.
	 \end{enumerate}
\end{example}

Let $a\in\mathbb{F}_q^*$ satisfy the two conditions in
Lemma~\ref{lem-binomial-irreducible} with $s=t+1$. Then
$
H(x)=x^{t+1}-a
$
is irreducible over $\mathbb{F}_q$. Substituting such H(x) into Theorem \ref{thm-2} yields the following explicit
class of Goppa polynomials and the  corresponding codes
attain their designed distance.

\begin{remark}\label{rem-binomial}
	Let $a\in\mathbb{F}_q^*$ satisfy the conditions in
	Lemma~\ref{lem-binomial-irreducible} with $s=t+1$. Then
	$
	H(x)=x^{t+1}-a
	$
	is irreducible over $\mathbb{F}_q$. Assume that its roots are
	$\alpha_1,\ldots,\alpha_{t+1}\in L$. By Theorem \ref{thm-2}, the following
	 family of Goppa polynomials gives Goppa codes $\Gamma(L,G)$ with designed distance:
	\[
	G(x)=\omega \sum_{i=1}^{t+1}\lambda_i\beta_i
	\frac{x^{t+1}-a}{x-\alpha_i},
	\quad
	\omega \in \Bbb F_{q^m}^*,
	\quad
	\lambda_i,\beta_i\in\mathbb{F}_q^*,
	\quad \text{and} \quad
	\sum_{i=1}^{t+1}\lambda_i\beta_i\neq 0.
	\]
	Provided that $G(\eta)\neq 0$ for every $\eta\in L$,
	we have
	$
	d(\Gamma(L,G))=t+1.
	$
\end{remark}

\begin{example}
	We give several computational examples arising from
	Remark~\ref{rem-binomial}. Let
	$
	s=t+1.
	$
	For each row, the polynomial
	$
	H(x)=x^s-a
	$
	is irreducible over $\mathbb F_q$ by
	Lemma~\ref{lem-binomial-irreducible}. Let $\alpha$ be a root of $H(x)$ and
	order the roots as
	$
	\alpha_i=\alpha^{q^{i-1}}
	$ for $1\le i\le s$.
	For each row in Table \ref{tab-examples-binomial-construction}, the Goppa polynomial is defined by
	$$
	G(x)
	=
	\omega\sum_{i=1}^{s}\lambda_i\beta_i
	\frac{H(x)}{x-\alpha_i}
	\in \mathbb F_{q^m}[x]
	$$
	and is omitted from the table since they may become lengthy.
The support $L$ is chosen to contain all roots of $H(x)$ and
satisfies $G(\eta)\neq 0$ for all $\eta\in L$. Furthermore, once $G(x)$ is
	fixed, the length of $L$ can vary from $\deg H$ up to
	$
	q^m-|\{\eta\in\mathbb F_{q^m}:G(\eta)=0\}|.
	$
	Hence, different choices of $L$ may lead to different Goppa codes $\Gamma(L,G)$. All the parameters of the corresponding Goppa codes
	$\Gamma(L,G)$ in Table \ref{tab-examples-binomial-construction} are computed by Magma. In particular, we write
	$
	\mathbb F_4=\mathbb F_2(\theta)
	$
	with
	$
	\theta^2+\theta+1=0
	$
	and
	$
	\mathbb F_9=\mathbb F_3(\theta)
	$
	with
	$
	\theta^2+1=0
	$
	for the rows with $q=4$ and $q=9$, respectively.
	
		\begin{table}[htbp]
		\centering
		\caption{Examples from the irreducible-binomial construction in Remark~\ref{rem-binomial}}
		\label{tab-examples-binomial-construction}
		\renewcommand{\arraystretch}{1.2}
		\begin{tabular}{|c|c|c|c|c|c|c|c|c|}
			\hline
			$q$ & $t$ & $m$ & $H(x)$ & $\omega$
			& $(\lambda_i)_{i=1}^{s}$ & $(\beta_i)_{i=1}^{s}$
			& $n$ & $[n,k,d]$\\
			\hline
			
			\multirow{2}{*}{$3$}
			& $1$ & $2$ & $x^2-2$
			& $\alpha+1$
			& $(1,1)$
			& $(1,1)$
			& $6$
			& $[6, 4, 2]$
			\\
			\cline{2-9}

			& $1$ & $2$ & $x^2-2$
			& $2\alpha+1$
			& $(2,1)$
			& $(1,2)$
			& $7$
			& $[7,5,2]$
			\\
			\hline

			\multirow{2}{*}{$4$}
			& $2$ & $3$
			& $x^3-\theta$
			& $\alpha+\theta$
			& $(1,\theta,\theta+1)$
			& $(1,1,\theta)$
			& $20$
			& $[20,14,3]$
			\\
			\cline{2-9}
			
			& $2$ & $3$
			& $x^3-(\theta+1)$
			& $\alpha+\theta+1$
			& $(\theta,\theta+1,1)$
			& $(1,\theta+1,\theta)$
			& $60$
			& $[60,54,3]$
			\\
			\hline

			\multirow{3}{*}{$5$}
			& $1$ & $2$
			& $x^2-3$
			& $2$
			& $(1,3)$
			& $(4,1)$
			& $10$
			& $[10,8,2]$
			\\
			\cline{2-9}

			& $3$ & $4$
			& $x^4-2$
			& $2$
			& $(1,2,1,3)$
			& $(1,1,1,1)$
			& $27$
			& $[27,15,4]$
			\\
			\cline{2-9}

			& $3$ & $4$
			& $x^4-3$
			& $3$
			& $(1,4,2,3)$
			& $(2,1,3,2)$
			& $53$
			& $[53,41,4]$
			\\
			\hline

			\multirow{4}{*}{$7$}
			& $2$ & $3$
			& $x^3-4$
			& $\alpha+2$
			& $(2,5,6)$
			& $(3,1,3)$
			& $50$
			& $[50,44,3]$
			\\
			\cline{2-9}
			
			& $2$ & $3$
			& $x^3-2$
			& $3$
			& $(1,4,2)$
			& $(2,3,1)$
			& $123$
			& $[123,117,3]$
			\\
			\cline{2-9}
			
		    & $5$ & $6$ & $x^6-5$ & $\alpha+2$
			& $(2,3,1,5,6,1)$ & $(3,1,3,2,3,4)$ & $63$ & $[63,33,6]$
			\\
			\cline{2-9}
			
			& $5$ & $6$
			& $x^6-3$
			& $2\alpha+1$
			& $(1,2,4,3,5,6)$
			& $(2,1,3,4,1,6)$
			& $85$
			& $[85,55,6]$
			\\
			\hline

			\multirow{2}{*}{$9$}
			& $3$ & $4$
			& $x^4-(\theta+1)$
			& $\alpha+2$
			& $(1,\theta,2,\theta+1)$
			& $(\theta+1,1,\theta,2+\theta)$
			& $50$
			& $[50,38,4]$
			\\
			\cline{2-9}
			
			& $3$ & $4$
			& $x^4-(2+\theta)$
			& $\alpha+\theta+1$
			& $(\theta,\theta+1,2,2+2\theta)$
			& $(1,2+\theta,\theta,\theta+1)$
			& $80$
			& $[80,68,4]$
			\\
			\hline
			
		\end{tabular}
	\end{table}

		In each row of Table~\ref{tab-examples-binomial-construction}, one checks that
	$
	\sum_{i=1}^{s}\lambda_i\beta_i\neq 0.
	$
	Hence Theorem~\ref{thm-2} applies, and the corresponding Goppa code satisfies
	$
	d(\Gamma(L,G))=\deg H.
	$
\end{example}

\begin{remark}\label{rem-trinomial}
	By Lemma~\ref{lem-trinomial-irreducible},
	$
	H(x)=x^p-x-a
	$
	is irreducible over $\mathbb{F}_q$ if $\operatorname{Tr}_{\mathbb{F}_q/\mathbb{F}_p}(a)\neq 0$ for $a\in\mathbb{F}_q$. Assume that the roots of $H(x)$ are
	$
	\alpha_1,\alpha_2,\ldots,\alpha_p\in L.
	$
	Substituting this choice of $H(x)$ into the Theorem \ref{thm-2} gives the
	following family of Goppa polynomials:
	\[
	G(x)=\omega \sum_{i=1}^{p}\lambda_i\beta_i
	\frac{x^p-x-a}{x-\alpha_i},
	\qquad \omega \in \Bbb F_{q^m}^*, \quad
	\lambda_i,\beta_i\in\mathbb{F}_q^*,
	\quad \text{and}\quad
	\sum_{i=1}^{p}\lambda_i\beta_i\neq 0.
	\]
	Provided that $G(\eta)\neq 0$ for every $\eta\in L$,
	the corresponding Goppa codes satisfy
	$
	d(\Gamma(L,G))=p.
	$
\end{remark}

\begin{example}
	We give several computational examples arising from
	Remark~\ref{rem-trinomial}.
     For each row in
	Table~\ref{tab-examples-trinomial-construction}, the polynomial
	$
	H(x)=x^p-x-a
	$
	is irreducible over $\mathbb F_q$ by
	Lemma~\ref{lem-trinomial-irreducible}. Let $\alpha$ be a root of $H(x)$ and
	order the roots as
	$
	\alpha_i=\alpha^{q^{i-1}}
	$ for $1\le i\le p$.
	For the chosen values of $\omega,\lambda_i$ and $\beta_i$, we construct
	$$
	G(x)=\omega\sum_{i=1}^{p}\lambda_i\beta_i
	\frac{H(x)}{x-\alpha_i}.
	$$
	The support $L$ is chosen to contain all roots of $H(x)$ and
satisfies $G(\eta)\neq 0$ for all $\eta\in L$. The parameters of the corresponding Goppa codes
	$\Gamma(L,G)$ in Table \ref{tab-examples-trinomial-construction} are computed by Magma.
	In particular, we write
	$
	\mathbb F_4=\mathbb F_2(\theta)$ with $
	\theta^2+\theta+1=0
	$ 
	and $
	\mathbb F_8=\mathbb F_2(\theta)$ with $
	\theta^3+\theta+1=0
	$ 
	for the row with $q=4$ and $q=8$, respectively.

	\begin{table}[htbp]
		\centering
		\caption{Examples from the irreducible-trinomial construction in Remark \ref{rem-trinomial}}
		\label{tab-examples-trinomial-construction}
		\renewcommand{\arraystretch}{1.2}
			\begin{tabular}{|c|c|c|c|c|c|c|c|c|c|c|}
				\hline
				$q$ & $p$ & $t$ & $m$ & $a$ & $H(x)$ & $\omega$
				& $(\lambda_i)_{i=1}^{p}$ & $(\beta_i)_{i=1}^{p}$& n & $[n,k,d]$\\
				\hline
				$3$ & $3$ & $2$ & $3$ & $1$
				& $x^3-x-1$
				& $\alpha+1$
				& $(1,1,2)$
				& $(1,1,1)$ & 8
				& $[8,2,3]$\\
				\hline
				$4$ & $2$ & $1$ & $2$ & $\theta$
				& $x^2+x+\theta$
				& $\alpha+1$
				& $(1,1)$
				& $(1,\theta)$ & 13
				& $[13,11,2]$\\
				\hline
				$5$ & $5$ & $4$ & $5$ & $1$
				& $x^5-x-1$
				& $2$
				& $(1,2,3,4,1)$
				& $(1,2,1,3,2)$& 120
				& $[120,100,5]$\\
				\hline
				$7$ & $7$ & $6$ & $7$ & $1$
				& $x^7-x-1$
				& $\alpha+2$
				& $(1,2,3,4,5,6,1)$
				& $(1,1,1,1,1,1,1)$& 69
				& $[69,27,7]$\\
				\hline
				$8$ & $2$ & $1$ & $2$ & $1$ & $x^2+x+1$
				& $\alpha+1$ & $(1,1)$ & $(1,\theta)$ & $38$ & $[38,36,2]$\\
				\hline
			\end{tabular}
	\end{table}
	
	In each row of Table~\ref{tab-examples-trinomial-construction}, one checks that
	$
	\sum_{i=1}^{p}\lambda_i\beta_i\neq 0.
	$
	Hence Theorem~\ref{thm-2} applies, and the corresponding Goppa code satisfies
	$
	d(\Gamma(L,G))=p=\deg H.
	$
\end{example}

\section{Explicit Families of Goppa Codes from the Auxiliary Polynomial $V(x)$}\label{sec-V}
We follow the notation in Section~\ref{sec-criterion} and assume that
$t+1\mid m$. Let $H(x)\in\mathbb F_q[x]$ be a monic irreducible
polynomial with $\deg(H)=t+1$.
Let
$
L=\{\alpha_1,\ldots,\alpha_n\}\subseteq\mathbb F_{q^m}
$
be the support of the Goppa code, and suppose that all roots of $H(x)$
are contained in $L$. We further assume that the Goppa polynomial
$G(x)$ is given by \eqref{eq-Gx}. 
The criterion in Theorem~\ref{thm-1} shows that the exact minimum distance
of $\Gamma(L,G)$ can be determined by the values of the auxiliary
polynomial $V(x)$ at the roots of $H(x)$. In this section, we use this criterion to obtain explicit families
of Goppa codes that attain their designed distance by choosing suitable
auxiliary polynomials $V(x)$.

We first consider a monomial
$
V(x)=\lambda x^\mu .
$
By imposing some conditions on the exponent $\mu$, the values of
$V(x)$ at the roots of $H(x)$ satisfy 
\eqref{eq-v-condition} in Theorem~\ref{thm-1}. Moreover, a trace condition
ensures that $\deg(G)=t$.
Consequently, the associated Goppa code attains its designed distance. This leads to the following monomial construction.

\begin{theorem}\label{thm-monomial-V-general}
	Let $H(x)\in\mathbb{F}_q[x]$ be a monic irreducible polynomial of degree $t+1$ and $\alpha\in\mathbb{F}_{q^{t+1}}^*$ be a root of $H(x)$. Then
	$
	H(x)=\prod_{i=0}^{t}\left(x-\alpha^{q^i}\right).
	$
	Write $\alpha_i=\alpha^{q^{i-1}}$ for $1\leq i\leq t+1$, and assume that
	$\alpha_1,\alpha_2,\ldots,\alpha_{t+1}\in L$.
	Let
	$
	V(x)=\lambda x^\mu
	$,
	where $\lambda\in\mathbb{F}_{q^m}^*$ and $\mu>0$. Let $G(x)$ be the unique
	remainder of $V(x)H'(x)$ modulo $H(x)$,  which is given by \eqref{eq-Gx}.
	Assume that $G(\eta)\neq0$ for every $\eta\in L$. Set
	$
	N_\alpha=\operatorname{ord}(\alpha).
	$
	If
	$$
	\operatorname{Tr}_{\mathbb{F}_{q^{t+1}}/\mathbb{F}_q}(\alpha^\mu)\neq0
	$$
	and
	$$
	\frac{N_\alpha}{\gcd\left(N_\alpha,(q-1)^2\right)}\mid \mu,
	$$
	then 
	$
	d(\Gamma(L,G))=t+1.
	$
\end{theorem}

\begin{proof}
	Note that
	$
	G(x)=U(x)H(x)+V(x)H'(x).
	$
	Then, for each root $\alpha_i$ of $H(x)$, we have
	$$
	G(\alpha_i)=V(\alpha_i)H'(\alpha_i),
	\qquad 1\leq i\leq t+1.
	$$
	
	We first prove that $\deg(G)=t$. Since $\deg(G)\leq t$ and
	$G(\alpha_i)=V(\alpha_i)H'(\alpha_i)$ for $1\leq i\leq t+1$, the Lagrange
	interpolation formula gives
	$$
	G(x)
	=
	\sum_{i=1}^{t+1}
	G(\alpha_i)
	\frac{H(x)}{(x-\alpha_i)H'(\alpha_i)}.
	$$
	Substituting $G(\alpha_i)=V(\alpha_i)H'(\alpha_i)$ into the above identity,
	we obtain
	$$
	G(x)
	=
	\sum_{i=1}^{t+1}
	V(\alpha_i)\frac{H(x)}{x-\alpha_i}.
	$$
	Since $V(x)=\lambda x^\mu$ and $\alpha_i=\alpha^{q^{i-1}}$, we have
	$$
	G(x)
	=
	\lambda \sum_{i=1}^{t+1}
	\alpha_i^\mu
	\frac{H(x)}{x-\alpha_i}.
	$$
	Each polynomial $H(x)/(x-\alpha_i)$ is monic of degree $t$. Hence the
	coefficient of $x^t$ in $G(x)$ is
	$
	\lambda \sum_{i=1}^{t+1}\alpha_i^\mu.
	$
	Note that $\alpha_i=\alpha^{q^{i-1}}$. Then we have
	$$
	\sum_{i=1}^{t+1}\alpha_i^\mu
	=
	\sum_{i=0}^{t}(\alpha^\mu)^{q^i}
	=
	\operatorname{Tr}_{\mathbb{F}_{q^{t+1}}/\mathbb{F}_q}(\alpha^\mu).
	$$
	By assumption,
	$
	\operatorname{Tr}_{\mathbb{F}_{q^{t+1}}/\mathbb{F}_q}(\alpha^\mu)\neq0.
	$
	Since $\lambda\neq0$, the coefficient of $x^t$ in $G(x)$ is nonzero and
	$
	\deg(G)=t.
	$
	
	It remains to verify \eqref{eq-v-condition} in
	Theorem~\ref{thm-1}. For $1\leq i\leq t$, we have
	$$
	\frac{V(\alpha_{t+1})}{V(\alpha_i)}
	=
	\frac{\lambda\alpha^{\mu q^t}}
	{\lambda\alpha^{\mu q^{i-1}}}
	=
	\alpha^{\mu(q^t-q^{i-1})}.
	$$
	We prove that this element belongs to $\mathbb{F}_q^*$. Since
	$N_\alpha=\operatorname{ord}(\alpha)$, we have
	$
	\alpha^k\in\mathbb{F}_q^*
	$
	if and only if
	$
	(\alpha^k)^q=\alpha^k
	$ for any integer $k$,
	which is equivalent to
	$$
	N_\alpha\mid k(q-1).
	$$
	Thus it suffices to prove that
	$$
	N_\alpha\mid \mu(q^t-q^{i-1})(q-1),
	\qquad 1\leq i\leq t.
	$$
	Since $\gcd(N_\alpha,q)=1$, this is equivalent to
	$$
	N_\alpha\mid \mu(q^{t-i+1}-1)(q-1),
	\qquad 1\leq i\leq t.
	$$
	Let $j=t-i+1$. Then it suffices to show that
	$$
	N_\alpha\mid \mu(q^j-1)(q-1),
	\qquad 1\leq j\leq t.
	$$
	By assumption,
	$$
	\frac{N_\alpha}{\gcd(N_\alpha,(q-1)^2)}\mid \mu.
	$$
	Therefore,
	$$
	N_\alpha\mid \mu(q-1)^2.
	$$
	For every $1\leq j\leq t$, we have
	$
	q^j-1=(q-1)(1+q+\cdots+q^{j-1}).
	$
	Hence
	$$
	N_\alpha\mid \mu(q^j-1)(q-1),
	\qquad 1\leq j\leq t.
	$$
	Consequently,
	$$
	\frac{V(\alpha_{t+1})}{V(\alpha_i)}
	\in\mathbb{F}_q^*,
	\qquad 1\leq i\leq t.
	$$
	By Theorem~\ref{thm-1}, the Goppa code $\Gamma(L,G)$ attains its designed
	distance, namely
	$
	d(\Gamma(L,G))=t+1.
	$
	This completes the proof.
\end{proof}

Theorem~\ref{thm-monomial-V-general} allows many choices of the irreducible
polynomial $H(x)$. Each such choice may give rise to different Goppa
polynomials and different Goppa codes, provided that the corresponding
remainder $G(x)$ has degree $t$ and
satisfies $G(\eta)\neq 0$ for all $\eta\in L$.
We now give a concrete family by taking $H(x)$ to be an irreducible binomial as in Lemma~\ref{lem-binomial-irreducible}.

\begin{remark}\label{rem-v-monomial}
	Let $a\in\mathbb{F}_q^*$ satisfy the conditions in
	Lemma~\ref{lem-binomial-irreducible} with $s=t+1$, and set
	$
	H(x)=x^{t+1}-a.
	$
	Then $H(x)$ is irreducible over $\mathbb{F}_q$. Suppose that the roots of
	$H(x)$ are $\alpha_1,\alpha_2,\ldots,\alpha_{t+1}$, and assume that
	$\alpha_1,\alpha_2,\ldots,\alpha_{t+1}\in L$.

	Let $\lambda\in\mathbb{F}_{q^m}^*$ and $\mu=\ell(t+1)$ for some integer
	$\ell\geq 1$. For $V(x)=\lambda x^\mu$, since
	$
	H'(x)=(t+1)x^t
	$
	and
	$
	x^{t+1}\equiv a\pmod{H(x)},
	$
	the unique remainder of $V(x)H'(x)$ modulo $H(x)$ is
	$$
	G(x)=\lambda(t+1)a^\ell x^t.
	$$
	By Lemma~\ref{lem-binomial-irreducible}, the characteristic of
	$\mathbb{F}_q$ does not divide $t+1$. Hence $t+1\neq0$ in $\mathbb{F}_q$.
	Since $\lambda\neq0$ and $a\neq0$, we have
	$
	\deg(G)=t.
	$
	
	Moreover, for every $1\leq i\leq t+1$, we have $\alpha_i^{t+1}=a$. Hence
	$$
	V(\alpha_i)
	=
	\lambda\alpha_i^\mu
	=
	\lambda(\alpha_i^{t+1})^\ell
	=
	\lambda a^\ell,\qquad 1\leq i \leq t+1.
	$$
	Therefore, 
	$$
	\frac{V(\alpha_{t+1})}{V(\alpha_i)}
	=
	1\in\mathbb{F}_q^*,
	\qquad 1\leq i\leq t.
	$$
	Provided that $G(\eta)\neq0$ for every $\eta\in L$, Theorem~\ref{thm-1}
	implies that
	$
	d\left(\Gamma(L,G)\right)=t+1.
	$
\end{remark}

\begin{example}
	We give several examples arising from Remark~\ref{rem-v-monomial}. Let
	$
	H(x)=x^{t+1}-a
	$
	be irreducible over $\mathbb F_q$ as in Lemma~\ref{lem-binomial-irreducible}.
	Let $\alpha$ be a root of $H(x)$, and
	write
	$
	\alpha_i=\alpha^{q^{i-1}}
	$
	for $1\leq i\leq t+1$. 
	For each row in Table~\ref{tab-monomial-vx-construction}, we choose
	$
	\mu=\ell(t+1)
	$, $V(x)=\lambda x^\mu$,
	and construct
	$
	G(x)=\lambda(t+1)a^\ell x^t \in \mathbb F_{q^m}[x].
	$
	The support $L$ is chosen such that it contains all roots of $H(x)$ and
	satisfies $G(\eta)\neq 0$ for all $\eta\in L$. The parameters of the corresponding Goppa codes
	$\Gamma(L,G)$ in Table \ref{tab-monomial-vx-construction} are computed by Magma, and $d\left(\Gamma(L,G)\right)=t+1$.
	In particular, we write
	$
	\mathbb F_4=\mathbb F_2(\theta)$ and $
	\theta^2+\theta+1=0
	$ for the row with $q=4$.
	
\begin{table}[htbp]
	\centering
	\caption{Monomial construction for $V(x)=\lambda x^\mu$ in Remark \ref{rem-v-monomial}}
	\label{tab-monomial-vx-construction}
	\renewcommand{\arraystretch}{1.2}
	\begin{tabular}{|c|c|c|c|c|c|c|c|c|}
		\hline
		$q$ & $t$ & $m$ & $H(x)=x^{t+1}-a$ & $\lambda$
		& $\ell$ & $G(x)$ & $n$ & $[n,k,d]$ \\
		\hline
		
		$3$ & $1$ & $2$ & $x^2-2$ & $\alpha$
		& $1$ & $\alpha x$ & $8$ & $[8,6,2]$ \\
		\hline
		
		$4$ & $2$ & $3$ & $x^3-(\theta+1)$ & $\alpha+1$
		& $2$ & $(\alpha+1)\theta x^2$ & $20$ & $[20,14,3]$ \\
		\hline
		
		$5$ & $3$ & $4$ & $x^4-2$ & $\alpha+2$
		& $1$ & $3(\alpha+2)x^3$ & $40$ & $[40,28,4]$ \\
		\hline
		
		\multirow{2}{*}{$7$}
	   & $2$ & $3$ & $x^3-3$ & $\alpha+2$
		& $2$ & $6(\alpha+2)x^2$ & $72$ & $[72,66,3]$ \\
		\cline{2-9}
		
		 & $5$ & $6$ & $x^6-5$ & $\alpha+2$
		& $2$ & $3(\alpha+2)x^5$ & $43$ & $[43,13,6]$ \\
		\hline
		
		$11$ & $4$ & $5$ & $x^5-2$ & $\alpha$
		& $2$ & $9\alpha x^4$ &$36$ & $[36,16,5]$ \\
		\hline
	
	\end{tabular}
\end{table}
\end{example}

We next consider the binomial
$
V(x)=\lambda_1x^\mu+\lambda_2.
$
By imposing suitable conditions on the exponent \(\mu\) and the coefficients
\(\lambda_1,\lambda_2\), we ensure that the values of \(V(x)\) at the conjugate
roots of \(H(x)\) satisfy \eqref{eq-v-condition} in Theorem \ref{thm-1} and $\deg(G) = t$.
Consequently, the corresponding Goppa code attains its designed
distance. This gives the following binomial construction.

\begin{theorem}\label{thm-binomial-V}
	Let $H(x)\in\mathbb{F}_q[x]$ be a monic irreducible polynomial of degree
	$t+1$ and $\alpha\in\mathbb{F}_{q^{t+1}}^*$ be a root of $H(x)$. Then
	$
	H(x)=\prod_{i=0}^{t}\left(x-\alpha^{q^i}\right).
	$
	Write $\alpha_i=\alpha^{q^{i-1}}$ for $1\leq i\leq t+1$, and assume that
	$\alpha_1,\alpha_2,\ldots,\alpha_{t+1}\in L$.
	Let
	$
	V(x)=\lambda_1x^\mu+\lambda_2,
	$
	where $\lambda_1,\lambda_2 \in\mathbb{F}_{q^m}^*$ and
	$\mu>0$. Let $G(x)$ be the unique remainder of $V(x)H'(x)$ modulo $H(x)$,
	which is given by \eqref{eq-Gx}. Assume that $G(\eta)\neq0$ for every $\eta\in L$. Set
	$
	N_\alpha=\operatorname{ord}(\alpha).
	$
	Suppose that
	$$
	\lambda_1
	\operatorname{Tr}_{\mathbb{F}_{q^{t+1}}/\mathbb{F}_q}(\alpha^\mu)
	+
	(t+1)\lambda_2
	\neq0,
	$$
	$$
	\frac{N_\alpha}{\gcd(N_\alpha,(q-1)^2)}\mid \mu,
	$$
	$$
	1+ \frac{\lambda_2}{\lambda_1\alpha^{\mu q^t}}\in\mathbb{F}_q^*,
	$$
	and 
	$$
	\lambda_2 \ne -\lambda_1\alpha^{\mu q^{i-1}} ,\qquad 1\le i \le t.
	$$
	 Then
	$
	d\left(\Gamma(L,G)\right)=t+1.
	$
\end{theorem}

\begin{proof}
	Note that
	$
	G(x)=U(x)H(x)+V(x)H'(x).
	$
	Then, for each root $\alpha_i$ of $H(x)$, we have
	$$
	G(\alpha_i)=V(\alpha_i)H'(\alpha_i),
	\qquad 1\leq i\leq t+1.
	$$
	
	We first prove that $\deg(G)=t$. Since $\deg(G)\leq t$ and
	$G(\alpha_i)=V(\alpha_i)H'(\alpha_i)$ for $1\leq i\leq t+1$, the Lagrange
	interpolation formula gives
	$$
	G(x)
	=
	\sum_{i=1}^{t+1}
	G(\alpha_i)
	\frac{H(x)}{(x-\alpha_i)H'(\alpha_i)}.
	$$
	Substituting $G(\alpha_i)=V(\alpha_i)H'(\alpha_i)$ into the above identity,
	we obtain
	$$
	G(x)
	=
	\sum_{i=1}^{t+1}
	V(\alpha_i)\frac{H(x)}{x-\alpha_i}.
	$$
	Each polynomial $H(x)/(x-\alpha_i)$ is monic of degree $t$. Hence the
	coefficient of $x^t$ in $G(x)$ is
	$
	\sum_{i=1}^{t+1}V(\alpha_i).
	$
	Since $V(x)=\lambda_1x^\mu+\lambda_2$, we have
	$$
	\sum_{i=1}^{t+1}V(\alpha_i)
	=
	\lambda_1\sum_{i=1}^{t+1}\alpha_i^\mu
	+
	(t+1)\lambda_2.
	$$
	Note that $\alpha_i=\alpha^{q^{i-1}}$. Then we have
	$$
	\sum_{i=1}^{t+1}\alpha_i^\mu
	=
	\sum_{i=0}^{t}(\alpha^\mu)^{q^i}
	=
	\operatorname{Tr}_{\mathbb{F}_{q^{t+1}}/\mathbb{F}_q}(\alpha^\mu).
	$$
	Therefore, the coefficient of $x^t$ in $G(x)$ is
	$
	\lambda_1
	\operatorname{Tr}_{\mathbb{F}_{q^{t+1}}/\mathbb{F}_q}(\alpha^\mu)
	+
	(t+1)\lambda_2.
	$
	By assumption, this coefficient is nonzero and
	$
	\deg(G)=t.
	$
	
	It remains to verify \eqref{eq-v-condition} in Theorem~\ref{thm-1}.
	For $1\leq i\leq t$, we have
	$$
	\frac{V(\alpha_{t+1})}{V(\alpha_i)}
	=
	\frac{\lambda_1\alpha^{\mu q^t}+\lambda_2}
	{\lambda_1\alpha^{\mu q^{i-1}}+\lambda_2}.
	$$
	Dividing both the numerator and the denominator by
	$\lambda_1\alpha^{\mu q^t}$ gives
	$$
	\frac{V(\alpha_{t+1})}{V(\alpha_i)}
	=
	\frac{
		1+\frac{\lambda_2}{\lambda_1\alpha^{\mu q^t}}
	}{
		\alpha^{\mu(q^{i-1}-q^t)}
		+
		\frac{\lambda_2}{\lambda_1\alpha^{\mu q^t}}
	}.
	$$
	We then prove that
	$$
	\alpha^{\mu(q^{i-1}-q^t)}\in\mathbb{F}_q^*,
	\qquad 1\leq i\leq t.
	$$
	Since $N_\alpha=\operatorname{ord}(\alpha)$, we have
	$
	\alpha^k\in\mathbb{F}_q^*
	$
	if and only if
	$
	(\alpha^k)^q=\alpha^k
	$ for any integer $k$,
	which is equivalent to
	$
	N_\alpha\mid k(q-1).
	$
	Thus, it is enough to prove that
	$$
	N_\alpha\mid \mu(q^{i-1}-q^t)(q-1),
	\qquad 1\leq i\leq t.
	$$
	Since $\gcd(N_\alpha,q)=1$, this is equivalent to
	$$
	N_\alpha\mid \mu(q^{t-i+1}-1)(q-1),
	\qquad 1\leq i\leq t.
	$$
	Let $j=t-i+1$. Then it suffices to show that
	$$
	N_\alpha\mid \mu(q^j-1)(q-1),
	\qquad 1\leq j\leq t.
	$$
	By assumption,
	$$
	\frac{N_\alpha}{\gcd(N_\alpha,(q-1)^2)}\mid \mu.
	$$
	Therefore, it is easy to know that
	$
	N_\alpha\mid \mu(q-1)^2.
	$
	For every $1\leq j\leq t$, we have
	$
	q^j-1=(q-1)(1+q+\cdots+q^{j-1}).
	$
	Hence,
	$$
	N_\alpha\mid \mu(q^j-1)(q-1),
	\qquad 1\leq j\leq t.
	$$
	Consequently,
	$$
	\alpha^{\mu(q^{i-1}-q^t)}\in\mathbb{F}_q^*,
	\qquad 1\leq i\leq t.
	$$
	By assumption,
	$$
	1+ \frac{\lambda_2}{\lambda_1\alpha^{\mu q^t}}\in\mathbb{F}_q^*
	$$
		and 
	$$
	\lambda_2 \ne -\lambda_1\alpha^{\mu q^{i-1}},\qquad 1\le i \le t.
	$$
	Therefore, we have
	$$
	\frac{V(\alpha_{t+1})}{V(\alpha_i)}
	\in\mathbb{F}_q^*,
	\qquad 1\leq i\leq t.
	$$
	By Theorem~\ref{thm-1}, the Goppa code $\Gamma(L,G)$ attains its designed
	distance, namely
	$
	d(\Gamma(L,G))=t+1.
	$
	This completes the proof.
\end{proof}

We next specialize Theorem~\ref{thm-binomial-V} to the case where
$H(x)=x^{t+1}-a$ is an irreducible binomial of the form considered in
Lemma~\ref{lem-binomial-irreducible}.

\begin{remark}\label{rem-binomial-V}
	Let $a\in\mathbb{F}_q^*$ satisfy the conditions in
	Lemma~\ref{lem-binomial-irreducible} with $s=t+1$, and set
	$
	H(x)=x^{t+1}-a.
	$
	Then $H(x)$ is irreducible over $\mathbb{F}_q$. Let $\alpha$ be a root of
	$H(x)$ and write
	$
	\alpha_i=\alpha^{q^{i-1}}
	$ for $ 1\leq i\leq t+1$.
	Assume that
	$
	\alpha_1,\alpha_2,\ldots,\alpha_{t+1}\in L.
	$
	
	Let
	$
	V(x)=\lambda_1x^\mu+\lambda_2,
	$
	where $\lambda_1,\lambda_2 \in\mathbb{F}_{q^m}^*$. Suppose that
	$
	\mu-1=\ell_1(t+1)+\ell_2
	$
	for some integers $\ell_1\geq 0$ and $0\leq \ell_2\leq t$ with
	$(\ell_1,\ell_2)\neq(0,0)$.
	Since
	$
	H'(x)=(t+1)x^t
	$
	and
	$
	x^{t+1}\equiv a\pmod{H(x)},
	$
	we have
	$$
	G(x)
	=
	\lambda_2 (t+1)x^t
	+
	\lambda_1(t+1)a^{\ell_1+1}x^{\ell_2}.
	$$
	By Lemma~\ref{lem-binomial-irreducible}, the characteristic of
	$\mathbb F_q$ does not divide $t+1$. Hence $t+1\neq0$ in $\mathbb F_q$.
	If $0\leq\ell_2<t$, then
	$
	\deg(G)=t
	$.
	If $\ell_2=t$, then
	$$
	G(x)
	=
	(t+1)(\lambda_2+\lambda_1a^{\ell_1+1})x^t,
	$$
	and 
	$
	\deg(G)=t
	$
	under the condition that
	$
	\lambda_2+\lambda_1a^{\ell_1+1}\neq0.
	$
	Set
	$
	N_\alpha=\operatorname{ord}(\alpha).
	$
	Assume further that
	$$
	\frac{N_\alpha}{\gcd\left(N_\alpha,(q-1)^2\right)}\mid \mu,
	$$
	$$
	1+\frac{\lambda_2}{\lambda_1\alpha^{\mu q^t}}\in\mathbb F_q^*,
	$$
	and
	$$
	\lambda_2\neq -\lambda_1\alpha^{\mu q^{i-1}},
	\qquad 1\leq i\leq t.
	$$
	If $G(\eta)\neq0$ for every $\eta\in L$, then
	Theorem~\ref{thm-binomial-V} implies that
	$
	d(\Gamma(L,G))=t+1.
	$
\end{remark}

\begin{example}
	We give several examples arising from Remark~\ref{rem-binomial-V}.
	Let
	$
	H(x)=x^{t+1}-a
	$
	be irreducible over $\mathbb F_q$ as in
	Lemma~\ref{lem-binomial-irreducible}. Let $\alpha$ be a root of $H(x)$ and
	write
	$
	\alpha_i=\alpha^{q^{i-1}}
	$
	for $1\leq i\leq t+1$. 
	For each row in Table~\ref{tab-binomial-vx-construction}, we take
	$
	V(x)=\lambda_1x^\mu+\lambda_2
	$
	and 
	$
	\mu-1=\ell_1(t+1)+\ell_2
	$
	with $\ell_1\geq0$ and $0\leq\ell_2\leq t$. The corresponding Goppa
	polynomial is
	$$
	G(x)
	=
	\lambda_2(t+1)x^t
	+
	\lambda_1(t+1)a^{\ell_1+1}x^{\ell_2}.
	$$
	In this table, $c\in \Bbb F_q^*$ is chosen so that
	$
	\lambda_2=c\lambda_1\alpha^{\mu q^t}
	$.
	The support $L$ is chosen to contain all roots of $H(x)$ and to satisfy
	$
	G(\eta)\neq0
	$
	for every $\eta\in L$. The parameters of the corresponding Goppa codes
	$\Gamma(L,G)$ in
	Table~\ref{tab-binomial-vx-construction} are computed by Magma. In particular, we write
	$
	\mathbb F_4=\mathbb F_2(\theta)$ and $
	\theta^2+\theta+1=0
	$ for the row with $q=4$.

\begin{table}[htbp]
	\centering
	\caption{Binomial construction for $V(x)=\lambda_1x^{\mu}+\lambda_2$ in Remark~\ref{rem-binomial-V}}
	\label{tab-binomial-vx-construction}
	\renewcommand{\arraystretch}{1.2}
	\begin{tabular}{|c|c|c|c|c|c|c|c|c|c|c|c|}
		\hline
		$q$ & $t$ & $m$ & $H(x)=x^{t+1}-a$
		& $\mu$ & $(\ell_1,\ell_2)$ & $c$
		& $\lambda_1$ & $\lambda_2$ & $G(x)$ & $n$ & $[n,k,d]$ \\
		\hline
		
		$4$ & $2$ & $3$ & $x^3-(\theta+1)$
		& $3$ & $(0,2)$ & $\theta$
		& $1$ & $1$ & $\theta x^2$
		& $20$ & $[20,14,3]$ \\
		\hline
		
		\multirow{2}{*}{$5$}
		 & $1$ & $2$ & $x^2-2$
		& $1$ & $(0,0)$ & $2$
		& $1$ & $3\alpha$ & $\alpha x+4$
		& $15$ & $[15,13,2]$ 
		\\
		\cline{2-12}

		 & $3$ & $4$ & $x^4-2$
		& $2$ & $(0,1)$ & $2$
		& $1$ & $3\alpha^2$ & $2\alpha^2x^3+3x$
		& $100$ & $[100,88,4]$ \\
		\hline

		\multirow{3}{*}{$7$}
		 & $2$ & $3$ & $x^3-3$
		& $2$ & $(0,1)$ & $1$
		& $1$ & $2\alpha^2$ & $6\alpha^2x^2+2x$
		& $48$ & $[48,42,3]$ \\
		\cline{2-12}
		
		 & $2$ & $3$ & $x^3-3$
		& $4$ & $(1,0)$ & $2$
		& $1$ & $5\alpha$ & $\alpha x^2+6$
		& $52$ & $[52,46,3]$ \\
		\cline{2-12}
		
	    & $5$ & $6$ & $x^6-5$
		& $2$ & $(0,1)$ & $1$
		& $1$ & $2\alpha^2$ & $5\alpha^2x^5+2x$
		& $72$ & $[72,42,6]$ \\
		\hline
		
			\multirow{2}{*}{$11$}
		 & $1$ & $2$ & $x^2-6$
		& $3$ & $(1,0)$ & $2$
		& $1$ & $10\alpha$ & $9\alpha x+6$
		& $42$ & $[42,40,2]$ \\
		\cline{2-12}
		
	 & $4$ & $5$ & $x^5-6$
		& $8$ & $(1,2)$ & $1$
		& $1$ & $10\alpha^3$ & $6\alpha^3x^4+4x^2$
		& $63$ & $[63,43,5]$ \\
		\hline
		
			\multirow{2}{*}{$13$}
		& $2$ & $3$ & $x^3-6$
		& $5$ & $(1,1)$ & $1$
		& $1$ & $2\alpha^2$ & $6\alpha^2x^2+4x$
		& $47$ & $[47,41,3]$ \\

		\cline{2-12}
		
		 & $3$ & $4$ & $x^4-6$
		& $2$ & $(0,1)$ & $2$
		& $1$ & $11\alpha^2$ & $5\alpha^2x^3+11x$
		& $79$ & $[79,67,4]$ \\
		\hline
	\end{tabular}
\end{table}
\end{example}

Motivated by the constructions of $V(x)$ in
Theorems~\ref{thm-monomial-V-general} and~\ref{thm-binomial-V},
we extend these results by considering products of the above
polynomials. The corresponding construction is given below.

\begin{theorem}\label{thm-product-monomial-binomial-V}
	Let $H(x)\in\mathbb{F}_q[x]$ be a monic irreducible polynomial of degree
	$t+1$ and $\alpha\in\mathbb{F}_{q^{t+1}}^*$ be a root of $H(x)$. Then
	$
	H(x)=\prod_{i=0}^{t}\left(x-\alpha^{q^i}\right).
	$
	Write $\alpha_i=\alpha^{q^{i-1}}$ for $1\leq i\leq t+1$, and assume that
	$\alpha_1,\alpha_2,\ldots,\alpha_{t+1}\in L$.
	Let
	$
	V_1(x)=\lambda_1x^{\mu_1}
	$
	and
	$
	V_2(x)=\lambda_2x^{\mu_2}+\lambda_3,
	$
	where
	$
	\lambda_1,\lambda_2,\lambda_3\in\mathbb{F}_{q^m}^*
	$
	and
	$
	\mu_1,\mu_2>0.
	$
	Set
	$
	V(x)=V_1(x)V_2(x)
	=
	\lambda_1x^{\mu_1}\left(\lambda_2x^{\mu_2}+\lambda_3\right).
	$
	Let $G(x)$ be the unique remainder of $V(x)H'(x)$ modulo $H(x)$,
	which is given by \eqref{eq-Gx}. Assume that $G(\eta)\neq0$ for every
	$\eta\in L$. Set
	$
	N_\alpha=\operatorname{ord}(\alpha).
	$
	Suppose that
	$$
	\lambda_1\lambda_2
	\operatorname{Tr}_{\mathbb{F}_{q^{t+1}}/\mathbb{F}_q}
	(\alpha^{\mu_1+\mu_2})
	+
	\lambda_1\lambda_3
	\operatorname{Tr}_{\mathbb{F}_{q^{t+1}}/\mathbb{F}_q}
	(\alpha^{\mu_1})
	\neq0,
	$$
	$$
	\frac{N_\alpha}{\gcd(N_\alpha,(q-1)^2)}\mid \mu_1,
	$$
	$$
	\frac{N_\alpha}{\gcd(N_\alpha,(q-1)^2)}\mid \mu_2,
	$$
	$$
	1+\frac{\lambda_3}{\lambda_2\alpha^{\mu_2q^t}}\in\mathbb F_q^*,
	$$
	and
	$$
	\lambda_3\neq -\lambda_2\alpha^{\mu_2q^{i-1}},
	\qquad 1\leq i\leq t.
	$$
	Then
	$
	d(\Gamma(L,G))=t+1.
	$
\end{theorem}

\begin{proof}
	Let
	$
	V_1(x)=\lambda_1x^{\mu_1}
	$, 
	$
	V_2(x)=\lambda_2x^{\mu_2}+\lambda_3
	$
	, and
	$
	V(x)=V_1(x)V_2(x) =
	\lambda_1\lambda_2x^{\mu_1+\mu_2}
	+
	\lambda_1\lambda_3x^{\mu_1}
	$.
	For each root
	$
	\alpha_i
	$
	of
	$
	H(x)
	$,
	we have
	$$
	G(\alpha_i)=V(\alpha_i)H'(\alpha_i),
	\qquad 1\leq i\leq t+1.
	$$
	
	We first prove that $\deg(G)=t$. Since $\deg(G)\leq t$, the Lagrange
	interpolation formula gives
	$$
	G(x)
	=
	\sum_{i=1}^{t+1}
	V(\alpha_i)\frac{H(x)}{x-\alpha_i}.
	$$
	Each polynomial $H(x)/(x-\alpha_i)$ is monic of degree $t$. Hence the
	coefficient of $x^t$ in $G(x)$ is
	$
	\sum_{i=1}^{t+1}V(\alpha_i).
	$
	It is easy to know that
	$$
	\sum_{i=1}^{t+1}V(\alpha_i)
	=
	\lambda_1\lambda_2
	\sum_{i=1}^{t+1}\alpha_i^{\mu_1+\mu_2}
	+
	\lambda_1\lambda_3
	\sum_{i=1}^{t+1}\alpha_i^{\mu_1}.
	$$
	Since $\alpha_i=\alpha^{q^{i-1}}$, we have
	$$
	\sum_{i=1}^{t+1}\alpha_i^{\mu_1+\mu_2}
	=
	\operatorname{Tr}_{\mathbb{F}_{q^{t+1}}/\mathbb{F}_q}
	(\alpha^{\mu_1+\mu_2})
	$$
	and
	$$
	\sum_{i=1}^{t+1}\alpha_i^{\mu_1}
	=
	\operatorname{Tr}_{\mathbb{F}_{q^{t+1}}/\mathbb{F}_q}
	(\alpha^{\mu_1}).
	$$
	Then the coefficient of $x^t$ in $G(x)$ is
	$$
	\lambda_1\lambda_2
	\operatorname{Tr}_{\mathbb{F}_{q^{t+1}}/\mathbb{F}_q}
	(\alpha^{\mu_1+\mu_2})
	+
	\lambda_1\lambda_3
	\operatorname{Tr}_{\mathbb{F}_{q^{t+1}}/\mathbb{F}_q}
	(\alpha^{\mu_1}).
	$$
	By assumption, this coefficient is nonzero and
	$
	\deg(G)=t.
	$
	
	It remains to verify \eqref{eq-v-condition} in Theorem~\ref{thm-1}. For
	$1\leq i\leq t$, we have
	$$
	\frac{V(\alpha_{t+1})}{V(\alpha_i)}
	=
	\frac{V_1(\alpha_{t+1})}{V_1(\alpha_i)}
		\frac{V_2(\alpha_{t+1})}{V_2(\alpha_i)}.
	$$
	By the same argument as in Theorem~\ref{thm-monomial-V-general}, the
	condition
	$$
	\frac{N_\alpha}{\gcd(N_\alpha,(q-1)^2)}\mid \mu_1
	$$
	implies that
	$$
	\frac{V_1(\alpha_{t+1})}{V_1(\alpha_i)}
	=
	\left(\frac{\alpha_{t+1}}{\alpha_i}\right)^{\mu_1}
	\in\mathbb F_q^*,
	\qquad 1\leq i\leq t.
	$$
	Next, applying the ratio argument in Theorem~\ref{thm-binomial-V} to
	$
	V_2(x)=\lambda_2x^{\mu_2}+\lambda_3
	$,
	the assumptions
	$$
	\frac{N_\alpha}{\gcd(N_\alpha,(q-1)^2)}\mid \mu_2,
	$$
	$$
	1+\frac{\lambda_3}{\lambda_2\alpha^{\mu_2q^t}}\in\mathbb F_q^*,
	$$
	and
	$$
	\lambda_3\neq -\lambda_2\alpha^{\mu_2q^{i-1}},
	\qquad 1\leq i\leq t,
	$$
	give that
	$$
	\frac{V_2(\alpha_{t+1})}{V_2(\alpha_i)}
	\in\mathbb F_q^*,
	\qquad 1\leq i\leq t.
	$$
	Therefore, we have
	$$
	\frac{V(\alpha_{t+1})}{V(\alpha_i)}
	\in\mathbb F_q^*,
	\qquad 1\leq i\leq t.
	$$
	Since $\deg(G)=t$ and $G(\eta)\neq0$ for every $\eta\in L$,
	Theorem~\ref{thm-1} implies that
	$$
	d(\Gamma(L,G))=t+1.
	$$
	This completes the proof.
\end{proof}

We next apply Theorem~\ref{thm-product-monomial-binomial-V} to the case where $H(x)=x^{t+1}-a$ is an irreducible binomial under the conditions
given in Lemma~\ref{lem-binomial-irreducible}.

\begin{remark}\label{rem-product-V}
	Let $a\in\mathbb{F}_q^*$ satisfy the conditions in
	Lemma~\ref{lem-binomial-irreducible} with $s=t+1$, and set
	$
	H(x)=x^{t+1}-a.
	$
	Then $H(x)$ is irreducible over $\mathbb{F}_q$. Let $\alpha$ be a root of
	$H(x)$ and write
	$
	\alpha_i=\alpha^{q^{i-1}}
	$
	for $1\leq i\leq t+1$. Assume that
	$
	\alpha_1,\alpha_2,\ldots,\alpha_{t+1}\in L.
	$
	
	Let
	$
	V_1(x)=\lambda_1x^{\mu_1}
	$
	and
	$
	V_2(x)=\lambda_2x^{\mu_2}+\lambda_3,
	$
	where
	$
	\lambda_1,\lambda_2,\lambda_3\in\mathbb{F}_{q^m}^*
	$
	and
	$
	\mu_1,\mu_2>0.
	$
	Set
	$
	V(x)=V_1(x)V_2(x).
	$
	Thus
	$
	V(x)
	=
	\lambda_1\lambda_2x^{\mu_1+\mu_2}
	+
	\lambda_1\lambda_3x^{\mu_1}.
	$
	Write
	$$
	\mu_1-1=\ell_1(t+1)+\ell_2,
	\qquad
	0\leq \ell_2\leq t,
	$$
	and
	$$
	\mu_1+\mu_2-1=\ell_3(t+1)+\ell_4,
	\qquad
	0\leq \ell_4\leq t,
	$$
	for some integers $\ell_1,\ell_3\geq0$. Since
	$
	H'(x)=(t+1)x^t
	$
	and
	$
	x^{t+1}\equiv a\pmod{H(x)},
	$
	we have
	\begin{align*}
		G(x)
		&\equiv V(x)H'(x)\pmod{H(x)}\\
		&=
		(t+1)\lambda_1\lambda_2a^{\ell_3+1}x^{\ell_4}
		+
		(t+1)\lambda_1\lambda_3a^{\ell_1+1}x^{\ell_2}.
	\end{align*}
	By Lemma~\ref{lem-binomial-irreducible}, the characteristic of
	$\mathbb{F}_q$ does not divide $t+1$. Hence $t+1\neq0$ in
	$\mathbb{F}_q$.
	If $\ell_4=t$ and $\ell_2<t$, then $\deg(G)=t$. If $\ell_2=t$ and
	$\ell_4<t$, then $\deg(G)=t$. If $\ell_2=\ell_4=t$, then
	$$
	G(x)
	=
	(t+1)\lambda_1
	\left(
	\lambda_2a^{\ell_3+1}
	+
	\lambda_3a^{\ell_1+1}
	\right)x^t,
	$$
	and hence $\deg(G)=t$ provided that
	$
	\lambda_2a^{\ell_3+1}
	+
	\lambda_3a^{\ell_1+1}
	\neq0.
	$
	
	Set
	$
	N_\alpha=\operatorname{ord}(\alpha).
	$
	Further assume that
	$$
	\frac{N_\alpha}{\gcd(N_\alpha,(q-1)^2)}\mid \mu_1 ,\quad
	\frac{N_\alpha}{\gcd(N_\alpha,(q-1)^2)}\mid \mu_2,
	$$
	$$
	1+\frac{\lambda_3}{\lambda_2\alpha^{\mu_2q^t}}\in\mathbb F_q^*,
	$$
	and
	$$
	\lambda_3\neq -\lambda_2\alpha^{\mu_2q^{i-1}},
	\qquad 1\leq i\leq t.
	$$
	If $G(\eta)\neq0$ for every $\eta\in L$, then
	Theorem~\ref{thm-product-monomial-binomial-V} implies that
	$
	d(\Gamma(L,G))=t+1.
	$
\end{remark}

\begin{example}\label{ex-product-V}
	We give several examples arising from Remark~\ref{rem-product-V}.
	Let
	$
	H(x)=x^{t+1}-a
	$
	be irreducible over $\mathbb F_q$ as in
	Lemma~\ref{lem-binomial-irreducible}. Let $\alpha$ be a root of $H(x)$, and
	write
	$
	\alpha_i=\alpha^{q^{i-1}}
	$
	for $1\leq i\leq t+1$. 
	For each row in Table~\ref{tab-product-v}, we take
	$
	V_1(x)=\lambda_1x^{\mu_1}
	$
	and
	$
	V_2(x)=\lambda_2x^{\mu_2}+\lambda_3,
	$
	where
	$
	\lambda_1,\lambda_2,\lambda_3\in\mathbb{F}_{q^m}^*
	$
	and
	$
	\mu_1,\mu_2>0.
	$
	The corresponding Goppa polynomial is
	$$
	G(x)
	=
	(t+1)\lambda_1\lambda_2a^{\ell_3+1}x^{\ell_4}
	+
	(t+1)\lambda_1\lambda_3a^{\ell_1+1}x^{\ell_2},
	$$ where $\max\{\ell_2,\ell_4\}=t$.
		The support $L$ is chosen to contain all roots of $H(x)$ and to satisfy
	$
	G(\eta)\neq0
	$
	for every $\eta\in L$. The parameters of the corresponding Goppa codes
	$\Gamma(L,G)$ are computed by Magma and are listed in
	Table~\ref{tab-product-v}.

\begin{table}[htbp]
	\centering
	\caption{Examples of the product construction in Remark~\ref{rem-product-V}}
	\label{tab-product-v}
	\renewcommand{\arraystretch}{1.2}
	\scriptsize
	\begin{tabular}{|c|c|c|c|c|c|c|c|c|c|c|c|c|}
		\hline
		$q$ & $t$ & $m$ & $H(x)$ & $(\mu_1,\mu_2)$
		& $\lambda_1$ & $\lambda_2$ & $\lambda_3$
		& $(\ell_1,\ell_2)$ & $(\ell_3,\ell_4)$
		& $G(x)$ & $n$ & $[n,k,d]$ \\
		\hline
		
		\multirow{2}{*}{$5$}
		& $1$ & $2$ & $x^2-2$
		& $(2,3)$
		& $2\alpha$ & $3$ & $3\alpha$
		& $(0,1)$ & $(2,0)$
		& $3x+\alpha$
		& $20$ & $[20,18,2]$ \\
		\cline{2-13}

		 & $3$ & $4$ & $x^4-2$
		& $(4,2)$
		& $2\alpha$ & $3\alpha^2$ & $3$
		& $(0,3)$ & $(1,1)$
		& $3\alpha x^3+\alpha^3x$
		& $123$ & $[123,111,4]$ \\
		\hline
		
		\multirow{2}{*}{$7$}
		 & $2$ & $3$ & $x^3-3$
		& $(2,4)$
		& $3\alpha$ & $2\alpha^2$ & $4$
		& $(0,1)$ & $(1,2)$
		& $3x^2+3\alpha x$
		& $43$ & $[43,37,3]$ \\
		\cline{2-13}

		 & $5$ & $6$ & $x^6-3$
		& $(2,4)$
		& $5\alpha^2$ & $4\alpha$ & $2\alpha^5$
		& $(0,1)$ & $(0,5)$
		& $3\alpha^3x^5+\alpha x$
		& $72$ & $[72,42,6]$ \\
		\hline
		
		\multirow{1}{*}{$11$}

		 & $4$ & $5$ & $x^5-6$
		& $(5,8)$
		& $4\alpha$ & $7\alpha^2$ & $2$
		& $(0,4)$ & $(2,2)$
		& $9\alpha x^4+\alpha^3x^2$
		& $34$ & $[34,14,5]$ \\
		\hline
		
		\multirow{1}{*}{$13$}
		
				 & $2$ & $3$ & $x^3-6$
		& $(5,1)$
		& $3\alpha$ & $5\alpha^2$ & $11$
		& $(1,1)$ & $(1,2)$
		& $9x^2+2\alpha x$
		& $93$ & $[93,87,3]$ \\
		\hline
	\end{tabular}
\end{table}
\end{example}

The following remark gives a general class of polynomials $V(x)$ for which
\eqref{eq-v-condition} in Theorem~\ref{thm-1} is automatically satisfied.
Motivated by the fact that, when $q=2$, the ratio in \eqref{eq-v-condition}
necessarily equals
$
V(\alpha_{t+1})/V(\alpha_i)=1
$
for any $1\leq i\leq t$.

\begin{remark}
	Let $\alpha$ be a primitive element of $\mathbb{F}_{q^{t+1}}$, and let
	$
	H(x)=\prod_{i=0}^{t}\left(x-\alpha^{q^i}\right)\in\mathbb{F}_q[x].
	$
	Write $\alpha_i=\alpha^{q^{i-1}}$ for $1\leq i\leq t+1$, and assume that
	$\alpha_1,\alpha_2,\ldots,\alpha_{t+1}\in L$.
	Let
	$
	V(x)=\lambda x^\mu+x^\nu,
	$
	where $\lambda\in\mathbb{F}_q^*$ and $\mu,\nu>0$. Suppose that
	$
	\lambda\alpha^\mu+\alpha^\nu \in\mathbb{F}_q^*.
	$
	Then, for every $1\leq i\leq t+1$, we have
	$$
	V(\alpha_i)
	=
	\lambda\alpha_i^\mu+\alpha_i^\nu
	=
	\lambda(\alpha^\mu)^{q^{i-1}}+(\alpha^\nu)^{q^{i-1}}
	=
	(\lambda\alpha^\mu+\alpha^\nu)^{q^{i-1}}
	=
	\lambda\alpha^\mu+\alpha^\nu.
	$$
	Hence
	$$
	\frac{V(\alpha_{t+1})}{V(\alpha_i)}
	=
	1\in\mathbb{F}_q^*,
	\qquad 1\leq i\leq t.
	$$
	Thus \eqref{eq-v-condition} in Theorem~\ref{thm-1} is
	satisfied.

	Let $G(x)$ be the unique remainder of $V(x)H'(x)$ modulo $H(x)$, satisfying
	\eqref{eq-Gx}. Then
	$$
	G(x)
	=
	\sum_{i=1}^{t+1}V(\alpha_i)\frac{H(x)}{x-\alpha_i}.
	$$
	Since
	$
	V(\alpha_i)=\lambda\alpha^\mu+\alpha^\nu
	$
	for every $1\leq i\leq t+1$, we obtain
	$$
	G(x)
	=
	(\lambda\alpha^\mu+\alpha^\nu)
	\sum_{i=1}^{t+1}\frac{H(x)}{x-\alpha_i}.
	$$
	Therefore,
	$$
	G(x)
	=
	(\lambda\alpha^\mu+\alpha^\nu)H'(x),
	$$ since 	
	$
	H'(x)=\sum_{i=1}^{t+1}\frac{H(x)}{x-\alpha_i}.
	$
	
	In particular, if
	$
	H(x)=x^{t+1}-a
	$
	is an irreducible binomial satisfying Lemma~\ref{lem-binomial-irreducible},
	then
	$
	H'(x)=(t+1)x^t.
	$
	By Lemma~\ref{lem-binomial-irreducible}, the characteristic of
	$\mathbb{F}_q$ does not divide $t+1$. Hence $t+1\neq0$ in $\mathbb{F}_q$.
	Since
	$
	\lambda\alpha^\mu+\alpha^\nu\in\mathbb{F}_q^*,
	$
	we have
	$$
	G(x)
	=
	(t+1)(\lambda\alpha^\mu+\alpha^\nu)x^t,
	$$
	and hence
	$
	\deg(G)=t.
	$
	Provided that $G(\eta)\neq0$ for every $\eta\in L$, Theorem~\ref{thm-1}
	implies that
	$
	d(\Gamma(L,G))=t+1.
	$
\end{remark}

Besides taking products, one may also consider sums of the polynomials appearing in Theorems~\ref{thm-monomial-V-general}, \ref{thm-binomial-V}, and \ref{thm-product-monomial-binomial-V}. However, unlike the product case, \eqref{eq-v-condition} in Theorem \ref{thm-1} is not automatically preserved. For example, let $ V_1(x)=\lambda_1x^{\mu_1} $ be a monomial satisfying Theorem~\ref{thm-monomial-V-general}, and let $ V_2(x)=\lambda_2x^{\mu_2}+\lambda_3 $ be a binomial satisfying Theorem~\ref{thm-binomial-V}. 
Set
$$
V(x)=V_1(x)+V_2(x).
$$
For the sum $V(x)$, the ratio condition needs to be checked
separately. For $1\leq i\leq t$, write
$$
\frac{V_1(\alpha_{t+1})}{V_1(\alpha_i)}=r_i,
\qquad
\frac{V_2(\alpha_{t+1})}{V_2(\alpha_i)}=s_i,
$$
where $r_i,s_i\in\mathbb F_q^*$. Then
$$
\frac{V(\alpha_{t+1})}{V(\alpha_i)}=\frac{V_1(\alpha_{t+1})+V_2(\alpha_{t+1})}
{V_1(\alpha_i)+V_2(\alpha_i)}
=
\frac{r_iV_1(\alpha_i)+s_iV_2(\alpha_i)}
{V_1(\alpha_i)+V_2(\alpha_i)}.
$$
To ensure that this quotient belongs to $\mathbb F_q^*$, assume that
$$
u_i = \frac{V_1(\alpha_i)}{V_2(\alpha_i)}\in\mathbb F_q
\quad\text{and}\quad
V_1(\alpha_i)+V_2(\alpha_i)\neq0,
\qquad 1\leq i\leq t+1.
$$
Then 
$$
\frac{V(\alpha_{t+1})}{V(\alpha_i)}
=
\frac{r_iu_i+s_i}{u_i+1}
\in\mathbb F_q^*,
\qquad 1\leq i\leq t.
$$
Moreover, the coefficient of $x^t$ in the remainder of $V(x)H'(x)$ modulo
$H(x)$ is
$$
\sum_{i=1}^{t+1}V(\alpha_i)
=
\lambda_1
\operatorname{Tr}_{\mathbb F_{q^{t+1}}/\mathbb F_q}
(\alpha^{\mu_1})
+
\lambda_2
\operatorname{Tr}_{\mathbb F_{q^{t+1}}/\mathbb F_q}
(\alpha^{\mu_2})
+
(t+1)\lambda_3.
$$
Hence, to ensure that this remainder has degree $t$, we also assume that
$$
\lambda_1
\operatorname{Tr}_{\mathbb F_{q^{t+1}}/\mathbb F_q}
(\alpha^{\mu_1})
+
\lambda_2
\operatorname{Tr}_{\mathbb F_{q^{t+1}}/\mathbb F_q}
(\alpha^{\mu_2})
+
(t+1)\lambda_3
\neq0.
$$
Under these additional assumptions, the same argument as above gives a Goppa polynomial \(G(x)\) for which the corresponding Goppa code attains its designed distance \(t+1\).

\section{Summary and concluding remarks}\label{sec-sum}
In this paper, we characterized when a Goppa code attains its designed
distance. By exploiting the structure of Goppa polynomials of the
form
\[
G(x)=U(x)H(x)+V(x)H'(x),
\]
we transformed the existence of a minimum-weight codeword into an
algebraic condition on the values of the auxiliary polynomial $V(x)$ at
the roots of $H(x)$.

The main contributions of this paper are summarized as follows.

\begin{itemize}
	\item We developed a criterion for a 
	Goppa code to attain its designed distance (see Theorem \ref{thm-1}). More precisely, we proved that
	the code contains a codeword of weight $t+1$ on the roots of an
	irreducible polynomial $H(x)$ if and only if
	\[
	\frac{V(\alpha_{i_{t+1}})}{V(\alpha_{i_j})}\in\mathbb F_q^*,
	\qquad 1\leq i\leq t,
	\] 	where $\alpha_{i_1},\ldots,\alpha_{i_{t+1}}$ are the roots of $H(x)$.
	This criterion provides a direct characterization of the designed
	distance attainment.
	
	\item Based on the criterion in Theorem~\ref{thm-1}, we obtained a
	general family of Goppa codes that attain their designed distance by
	prescribing the values of $V(x)$ at the roots of $H(x)$ and applying
	Lagrange interpolation to determine the corresponding Goppa polynomials
	(see Theorem~\ref{thm-2}).
	
	\item We further derived explicit families of Goppa codes by considering
	monomial, binomial, and product forms of the auxiliary polynomial $V(x)$
	(see Theorems~\ref{thm-monomial-V-general},
	\ref{thm-binomial-V}, and
	\ref{thm-product-monomial-binomial-V}). The designed distance attainment
	is characterized by trace values, multiplicative orders, and
	nonvanishing constraints.
\end{itemize}

The results of this paper provide a general framework for characterizing Goppa codes with designed distance attained. Future work may focus on exploring additional structured forms of the Goppa polynomial $G(x)$ and determining the exact minimum distances of
the associated Goppa codes.

\end{document}